\documentclass{CSML}
\pdfoutput=1

\usepackage{lastpage}
\newcommand{\dOi}{14(1:21)2018}
\lmcsheading{}{1--\pageref{LastPage}}{}{}%
{Jan.~31, 2017}{Mar.~06, 2018}{}


\ACMCCS{[{\bf Theory of computation}]: Formal languages
and automata theory --- Automata over infinite objects ---
Quantitative automata}

\usepackage{microtype}
\usepackage{savesym,graphicx}
\usepackage{pgf, verbatim,eurosym}
\usepackage{sansmath}

\usepackage{url}
\usepackage{enumitem, stmaryrd,mathrsfs}
\usepackage{amsmath,latexsym, mathtools,amssymb}
\usepackage{times,xspace}

\usepackage{bbold, pgf, tikz, verbatim,graphicx}
\usetikzlibrary{arrows,automata}
\usepackage{microtype}
\newcommand*{\Sr}{S}
\newcommand*{\Cmc}{\mathcal{C}}
\newcommand*{\Pmc}{\mathcal{P}}

\newcommand*{\Snn}{\Sr^{n \times n}}
\newcommand*{\Spnn}{{\Sr'}^{n\times n}} 
\newcommand*{\Snngam}{(\Snn)^{\Gamma^* \times \Gamma^*}}
\newcommand*{\Mpp}{M_{p,p}}
\newcommand*{\Mppsq}{M_{p,p^2}}
\newcommand*{\Mpe}{M_{p,\epsilon}}

\newcommand*{\Mspe}{(M^*)_{p,\epsilon}}
\newcommand*{\Momk}{(M^{\omega,k})}

\newcommand*{\aim}{a_{im}}
\newcommand*{\aimi}{a_{im_1}}
\newcommand*{\bij}{b_{ij}}
\newcommand*{\cim}{c_{im}}
\newcommand*{\mpm}{[m_1,p,m_2]}
\newcommand*{\mpj}{[m,p,j]}
\newcommand*{\ipj}{[i,p,j]}
\newcommand*{\ip}{[i,p]}
\renewcommand{\mp}{[m,p]}

\newcommand*{\llss}{\ll \Sigma^* \gg}
\newcommand*{\llso}{\ll \Sigma^\omega \gg}
\newcommand*{\Nlse}{\N^\infty \langle \Sigma \cup \{\epsilon\} \rangle}

\newcommand*{\derl}{\Rightarrow_{\!L}}
\newcommand*{\derls}{\Rightarrow_{\!L}^*}
\newcommand*{\derlom}{\Rightarrow_{\!L}^{\omega,k}}

\renewcommand{\epsilon}{\varepsilon}

\newcommand{\N}{\mathbb{N}}
\newcommand{\B}{\mathbb{B}}

\newcommand*{\roc}{(\Sr^{n \times n})^{p^* \times p^*}}
\newcommand*{\bhvr}[1]{\Vert #1 \Vert}

\makeatletter
\newcommand*{\rom}[1]{\expandafter\@slowromancap\romannumeral #1@}
\makeatother

\usepackage{hyperref}
\hypersetup{hidelinks}
\begin{document}

\title{Weighted $\omega$-Restricted One-Counter Automata}

\author[M. Droste]{Manfred Droste}
\address{Universit\"at Leipzig,  Institut f\"ur Informatik}
\email{droste@informatik.uni-leipzig.de}
\thanks{This work was partially supported by DFG Graduiertenkolleg 1763 (QuantLA)}

\author[W. Kuich]{Werner Kuich}	
\address{Technische Universit\"at Wien, Institut f\"ur Diskrete Mathematik und Geometrie}
\email{kuich@tuwien.ac.at}  
\thanks{The second author was partially supported by Austrian Science Fund (FWF): grant no. I1661 – N25.}





\begin{abstract}
  \noindent Let $\Sr$ be a complete star-omega semiring and $\Sigma$ be an alphabet. 
  For a weighted $\omega$-restricted one-counter automaton $\Cmc$
  with set of states $\{1, \dots, n\}$, $n \geq 1$, we show that
  there exists a mixed algebraic system over a complete semiring-semimodule pair
  ${((\Sr \ll \Sigma^* \gg)^{n\times n},  (\Sr \ll \Sigma^{\omega}\gg)^n)}$ such
  that the behavior $\bhvr{\Cmc}$ of $\Cmc$ is a component of a solution of this system.
  In case the basic semiring is $\B$ or $\N^{\infty}$  we show that there exists
  a mixed context-free grammar that generates $\bhvr{\Cmc}$.
  The construction of the mixed context-free grammar from $\Cmc$ is a generalization of the well-known triple construction in case of restricted one-counter automata and is called now triple-pair construction for $\omega$-restricted one-counter automata.
\end{abstract}

\maketitle

\section{Introduction}
\label{section:Intro}

Restricted one-counter pushdown automata and languages were introduced by Greibach 
\cite{JACM16}
and considered in Berstel \cite{4}, Chapter \rom{7} 4. 
These restricted one-counter pushdown automata are pushdown automata having just one
pushdown symbol accepting by empty tape,
and the family of restricted one-counter languages is the family of languages
accepted by them.

Let L be the Lukasiewicz language, i.e., the formal language over the alphabet 
$\Sigma=\{a,b\}$ 
generated by the context-free grammar with productions $S \to aSS, S\to b$.
Then the family of restricted one-counter languages is the principal cone generated by L, while the family of one-counter languages is the full AFL generated by L.

All these results can be transferred to formal power series and restricted one-counter automata over them (see Kuich, Salomaa \cite{88}, Example 11.5).
Restricted one-counter automata can also be used to accept infinite words and it is this aspect we generalize in our paper.

We consider weighted $\omega$-restricted one-counter automata 
and their relation to algebraic systems over the complete 
semiring-semimodule pair $(\Sr ^{n\times n}, V^n)$,
where $\Sr$ is a complete star-omega semiring. 
It turns out that the well-known triple construction for pushdown
automata in case of unweighted restricted one-counter
automata can be generalized to a triple-pair construction for weighted 
$\omega$-restricted one-counter automata. In the classical theory, the triple construction yields for a given pushdown automaton an equivalent context-free grammar.
(See Harrison \cite{63}, Theorem 5.4.3; 
Bucher, Maurer \cite{17}, S\"atze 2.3.10, 2.3.30; Kuich, Salomaa \cite{88}, pages 178, 306; Kuich \cite{78}, page 642; 
\'Esik, Kuich \cite{MAT}, pages 77, 78.)

The paper consists of this and three more sections.
In Section \ref{sec:pre}, we review the necessary preliminaries.
In Section \ref{sec:2}, restricted one-counter matrices are introduced and their properties are studied. 
The main result is that, for such a matrix $M$, the $p$-block of the infinite column vector $M^{\omega,k}$ is a solution of the linear equation 
$z=(M_{p, p^2}(M^*)_{p,\epsilon} + M_{p,p}+ M_{p,p^2})z$.
In Section \ref{sec:3}, weighted $\omega$-restricted one-counter automata are introduced as a special case of weighted $\omega$-pushdown automata. 
We show that for a weighted $\omega$-restricted one-counter automaton $\Cmc$ there exists a mixed algebraic system such that the behavior $\bhvr{\Cmc}$ of $\Cmc$ is a component of a solution of this system.
In Section \ref{sec:4} we consider the case that the complete star-omega semiring $\Sr$ is equal to $\B$ or $\N^\infty $. 
Then for a given  weighted  $\omega$-restricted one-counter automaton $\Cmc$ a mixed context-free grammar is constructed that generates $\bhvr{\Cmc}$.
This construction is a generalization of the well-known triple construction in case of 
restricted one-counter automata and is called \emph{triple-pair construction} for $\omega$-restricted one-counter automata.

\section{Preliminaries}\label{sec:pre}

For the convenience of the reader, we quote definitions and results of \'Esik, Kuich \cite{42,43,45} from \'Esik, Kuich \cite{MAT}. The reader should be familiar with Sections 5.1-5.6 of \'Esik, Kuich \cite{MAT}.

A semiring $S$ is called \emph{complete} if it is possible to define sums for all
families $(a_i \mid i \in I)$ of elements of $S$, where $I$ is an arbitrary index set, such that the following conditions are satisfied 
(see Conway \cite{25}, Eilenberg \cite{29}, Kuich \cite{78}):
\begin{align*}
\text{(i)} \quad & \sum\limits_{i \in \emptyset} a_i = 0, \qquad 
\sum\limits_{i \in \{j\}} a_i = a_j, 
\qquad \sum\limits_{i \in \{j,k\}} a_i = a_j + a_k \text{ for } j \neq k \, , \\
\text{(ii)} \quad & \sum\limits_{j \in J}\big(\sum_{i \in I_j} a_i \big) = \sum_{i \in I} a_i \, , 
\text{ if }\ \bigcup_{j \in J}\! I_j = I \ \text{ and }\ I_j \cap I_{j'} = \emptyset \ 
\text{ for } \ j \neq j' \, ,\\
\text{(iii)} \quad & \sum_{i \in I} (c \cdot a_i) = c \cdot \big( \sum_{i \in I} a_i \big) ,
\qquad \sum_{i \in I} (a_i \cdot c) = \big( \sum_{i \in I } a_i \big) \cdot c \, .
\end{align*}

This means that a semiring $S$ is complete if it is possible to define ``infinite sums'' (i) that are an extension of the finite sums, (ii) that are associative and commutative and 
(iii) that satisfy the distribution laws.
If $S$ is a monoid and conditions (i) and (ii) are satisfied then $S$ is called a \emph{complete monoid}.

A semiring S equipped with an additional unary star operation $^*: \Sr \to \Sr$ is called a starsemiring.
In complete semirings for each element $a$, the \emph{star} $a^*$
of $a$ is defined by
\begin{equation*}
a^* = \sum_{j\geq 0} a^j \, .
\end{equation*}
Hence, each complete semiring is a starsemiring, called a \emph{complete starsemiring}.\par

Suppose that $\Sr$ is a semiring and $V$ is a commutative monoid written additively.
We call $V$ a (left) $\Sr$-semimodule if $V$ is 
equipped with a (left) action

\begin{align*}
\Sr \times V & \ \to\ V\\
(s,v) & \ \mapsto \  sv
\end{align*}
\noindent
subject to the following rules:

\begin{alignat*}{3}
s(s'v) & = (ss')v & 1v & = v\\
(s+s')v & = sv + s'v &  \hspace{2cm} 0v  & = 0 \\
s(v+v') & = sv + sv' & s0 & = 0,
\end{alignat*}

%
%
%
\noindent
for all $s,s' \in \Sr$ and $v,v' \in V$.
When V is an $\Sr$-semimodule, we call $(\Sr,V)$ a 
\emph{semiring-semimodule pair}.

Suppose that $(\Sr,V)$ is a semiring-semimodule pair 
such that $\Sr$ is a starsemiring and $\Sr$ and $V$ are equipped with an omega operation
$^\omega: \Sr \to V$.
Then we call $(\Sr,V)$ a \emph{starsemiring-omegasemimodule pair}.


\'Esik, Kuich \cite{44} define a 
\emph{complete semiring-semimodule pair} to be a semiring-semimodule pair $(\Sr,V)$ such that $\Sr$ is a complete semiring and V is a complete monoid with
\begin{align*}
s\big(\sum_{i\in I} v_i \big) & = \sum_{i \in I} sv_i\\
\big(\sum_{i\in I} s_i\big) v & = \sum_{i \in I} s_iv \, ,
\end{align*}
\noindent
for all $s\in \Sr$, $v \in V$, and for all families 
$(s_i)_{i \in I} $ over $\Sr$ and $(v_i)_{i \in I} $ over $V$;
moreover, it is required that an \emph{infinite product operation}
\begin{equation*}
(s_1, s_2, \ldots) \ \mapsto \ \prod_{j\geq 1} s_j
\end{equation*}
is given mapping infinite sequences over $\Sr$ to $V$ subject to the following three conditions:
\begin{align*}
\prod_{i\geq 1} s_i & \ = \ \prod_{i \geq 1} (s_{n_{i-1}+1}\cdot \dots \cdot s_{n_i})\\
s_1 \cdot \prod_{i \geq 1} s_{i+1} & \ = \ \prod_{i\geq 1} s_i\\
\prod_{j\geq 1} \sum_{i_j \in I_j} s_{i_j} & \ = \ 
\sum_{(i_1, i_2, \dots)\in I_1 \times I_2 \times \dots} \prod_{j\geq 1} s_{i_j} \, ,
\end{align*}

\noindent
where in the first equation $0=n_0 \leq n_1 \leq n_2 \leq \dots$ and $I_1, I_2, \dots$ are arbitrary index sets.
Suppose that $(\Sr, V)$ is complete.
Then we define
\begin{align*}
s^* & \ = \ \sum_{i\geq 0} s^i\\
s^\omega &\  = \ \prod_{i\geq 1} s \, ,
\end{align*}
\noindent
for all $s \in \Sr$. This turns $(\Sr, V)$ into a starsemiring-omegasemimodule pair.
Observe that, if $(\Sr,V)$ is a complete semiring-semimodule pair,
then $0^\omega = 0$.\par
For a starsemiring $\Sr$, we denote by $\Snn$ the semiring of $n\times n$-matrices over $\Sr$. 
If $(\Sr, V)$ is a complete semiring-semimodule pair then, by \'Esik, Kuich \cite{EK}, $(\Sr^{n \times n}, V^n)$ is again a complete semiring-semimodule pair.

A \emph{star-omega semiring} is a semiring $\Sr$ equipped with unary operations $^*$ and $^\omega : \Sr \to \Sr$.
A star-omega semiring $\Sr$ is called \emph{complete} if $(\Sr,\Sr)$ is a complete semiring semimodule pair, i.e., if $\Sr$ is complete and is equipped with an infinite product operation that satisfies the  three conditions stated above.
 For the theory of infinite words and finite automata accepting infinite words by the B\"uchi condition consult Perrin, Pin \cite{PerPin}.
 
\section{Restricted one-counter matrices}
\label{sec:2}

In this section we introduce restricted one-counter (roc) matrices. 
Restricted one-counter matrices are a special case of pushdown matrices introduced by Kuich, Salomaa \cite{88}.
A matrix $M \in \Snngam$ is termed
a \emph{pushdown transition matrix} (with \emph{pushdown alphabet} $\Gamma$ and \emph{set of states} $\{1,\dots,n\}$) if
\begin{enumerate}[label=(\roman*)]
        \item for each $p \in \Gamma$ there exist only finitely many blocks $M_{p,\pi}$, $\pi \in \Gamma^*$, that are non-zero;
        \item for all $\pi_1, \pi_2 \in \Gamma^*$,
                \begin{equation*}
                M_{\pi_1,\pi_2} = \left\{
                        \begin{array}{ll}
                        M_{p,\pi} & \hspace{0,2cm} \text{if there exist } p \in \Gamma, \pi,\pi' \in \Gamma^* \text{ with } \pi_1 = p\pi' \text{ and } \pi_2 = \pi \pi', \\
                        0 & \hspace{0,2cm} \text{otherwise.}
                        \end{array}
                        \right.
                \end{equation*}
\end{enumerate}
Theorem 10.5 of Kuich, Salomaa \cite{88} states that for pushdown matrices over power series semirings with particular properties, $(M^*)_{\pi_1 \pi_2, \varepsilon} = (M^*)_{\pi_1, \varepsilon} (M^*)_{\pi_2, \varepsilon}$ holds for all $\pi_1, \pi_2 \in \Gamma^*$. This result is generalized in the case of roc-matrices to arbitrary roc-matrices over complete starsemirings in Corollary \ref{cor:2.2}. Then we prove some important equalities for roc-matrices.
In Theorem \ref{thm:2.1} and Corollary \ref{cor:2.2}, $\Sr$ denotes a complete starsemiring; afterwards in this section, $(\Sr, V)$ denotes a complete semiring-semimodule pair.

A \emph{restricted one-counter} (abbreviated \emph{roc}) \emph{matrix (with counter symbol $p$)} is a matrix $M$
in $(\Snn)^{p^*\times p^*}$, for some $n\geq 1$, subject to the following condition:
There exist matrices $A,B,C \in \Snn$ such that,
for all $k \geq 1$, 
\begin{equation*}
M_{p^k, p^{k+1}}=A \, , \quad M_{p^k,p^k}=C \, \quad  M_{p^k,p^{k-1}}=B \, ,
\end{equation*}
and these blocks of $M$ are the only ones which
may be non-zero. (Here, $p^* = \{p^n \mid n \geq 0 \}$. A block of $M$ is an element of the matrix $M$ which is itself a matrix in $S^{n \times n}$.)

Observe that, for $k \geq 1$,
\begin{alignat*}{4}
M_{p^k,p^{k+1}}  & \ = \ M_{p, p^2} & & \ = \ A \, ,\\
M_{p^k,p^k} & \ = \  M_{p,p} & & \ = \ C \, , \\
M_{p^k,p^{k-1}}  & \ = \ M_{p, \epsilon} & &\ = \ B \, , \\
M_{\epsilon,p^k} & \ = \ M_{\epsilon, \epsilon} & & \ = \ 0 \, .
\end{alignat*}

Also note that the matrix $A$ (resp $B,C$) in $\Snn$ describes the weight of
transitions when pushing (resp., popping, not changing) an additional symbol $p$ to 
(resp., from) the pushdown counter.

\begin{thm}\label{thm:2.1}
	Let S be a complete starsemiring and $M$ be a roc-matrix.
	Then, for all $i\geq 0$, 
	\begin{equation*}
	(M^*)_{p^{i+1},\epsilon} \ = \ (M^*)_{p,\epsilon}(M^*)_{p^i,\epsilon} \, .
	\end{equation*}
\end{thm}

\begin{proof}
	First observe that
	\begin{equation*}
	(M^*)_{p^{i+1}\!,\epsilon} = \! \! \sum_{m\geq 0} (M^{m+1})_{p^{i+1}\!,\epsilon} = \! \!
	\sum_{m\geq 0} \hspace{-0.3cm} \sum_{\hspace{0.3cm} i_1,\dots, i_m \geq 1} \! \! \! \! \! M_{p^{i\text{+}1}\!,p^{i_1}} 
	M_{p^{i_1} \!,p^{i_2}} \dots M_{p^{i_{m\text{-}1}}\! ,p^{i_m}} M_{p^{i_m}\!,\epsilon} \, ,
	\end{equation*}
	where, for $m=0$, the product equals $M_{p^{i+1},\epsilon}$.
	Now we obtain
	
	\begin{align*}
	(M^*)_{p^{i+1},\epsilon}  =\ & 
	\sum_{m\geq 0}   \sum_{ i_1, \dots, i_m \geq 1}M_{p^{i+1},p^{i_1}} \dots
	M_{p^{i_{m-1}},p^{i_m}} M_{p^{i_m},\epsilon}\\
	= \ &  \sum_{m_1\geq 0} \big( \sum_{ j_1, \dots, j_{m_1} \geq 1}  M_{p^{i+1},p^{i+j_1}} \dots
	M_{p^{i+j_{m_1-1}},p^{i+j_{m_1}}} M_{p^{i+j_{m_1}},p^i}\big) \cdot\\
	& \sum_{m_2\geq 0}  \big( \sum_{ i_1, \dots, i_{m_2} \geq 1} 
	M_{p^{i},p^{i_1}} \dots
	M_{p^{i_{m_2-1}},p^{i_{m_2}}} M_{p^{i_{m_2}},\epsilon} \big) \\
	= \ & \sum_{m_1\geq 0} \big( \sum_{ j_1, \dots ,j_{m_1} \geq 1}  
	M_{p,p^{j_1}} \dots
	M_{p^{j_{m_1-1}},p^{j_{m_1}}} M_{p^{j_{m_1}},\epsilon}\big) 
	(M^*)_{p^i,\epsilon}\\
	= \ & (M^*)_{p,\epsilon} (M^*)_{p^i,\epsilon} \, .
	\end{align*}
	
	Clearly, in each sequence leading from $p^{i + 1}$ to $\varepsilon$, there is a first time at which the top $p$ is reduced to
	$\varepsilon$ and at which $p^i$ is seen. This moment is reached at the end of the second line.
	Hence, in the second line the pushdown contents 
	$p^{i+j_1}, \dots, p^{i+j_{m_1}}$, $m_1 \geq 0$ are always nonempty.
	
\end{proof}

\begin{cor}\label{cor:2.2}
	For all $i\geq 0$, 
	$\ (M^*)_{p^i,\epsilon} \ = \ ((M^*)_{p,\epsilon})^i $.
\end{cor}

\begin{lem}\label{lem:2.3}
	Let $(\Sr, V)$ be a complete semiring-semimodule pair.
	Let $M \in \roc$ be a roc-matrix.
	Then 
	$$(M^\omega)_{p^2}\ = \ (M^\omega)_p + (M^*)_{p,\epsilon}(M^\omega)_p \, .$$
\end{lem}

\begin{proof}
	Subsequently in the first equation we split the summation so that in the first summand there is no factor $M_{p^2,p}$, while in the second summand there is at least one factor $M_{p^2,p}$; since 
	$k_1, \dots, k_m \geq 2$, $M_{p^{k_m},p}$ is the first such factor.
	In the second equality we use the property of $M$ being a roc-matrix: $M_{p^i,p^j} = M_{p^{i-1},p^{j-1}}$ for $i \geq 2$,
	$j \geq 1$.	
	We compute:
	
	\begin{align*}
	(M^\omega)_{p^2}  = \ & \sum_{i_1,i_2,\dots \geq 2} M_{p^2,p^{i_1}}
	M_{p^{i_1},p^{i_2}} \dots + \\
	& \sum_{m\geq 0} \sum_{k_1, k_2, \dots, k_m \geq 2} \!\!\!
	M_{p^2,p^{k_1}}M_{p^{k_1},p^{k_2}} \dots M_{p^{k_m},p} 
	\sum_{j_1,j_2,\dots \geq 1} \! \! M_{p,p^{j_1}} M_{p^{j_1},p^{j_2}}\dots\\
	= \ & \sum_{i_1,i_2,\dots \geq 2} M_{p,p^{i_1-1}}
	M_{p^{i_1-1},p^{i_2-1}} \dots + \\
	& \sum_{m\geq 0} \sum_{k_1, k_2, \dots, k_m \geq 2} \!\!\!
	M_{p,p^{k_1-1}}M_{p^{k_1-1},p^{k_2-1}} \dots M_{p^{k_m-1},\epsilon} (M^\omega)_p\\
	& \ = (M^\omega)_p + \sum_{m \geq 0} (M^{m+1})_{p,\epsilon} (M^\omega)_p \\
	& \ = (M^\omega)_p + (M^*)_{p,\epsilon}(M^\omega)_p \, .
	\end{align*}
\end{proof}

Intuitively, our next theorem states that infinite computations starting with $p$ on the pushdown tape yield the same matrix $(M^\omega)_p$ as the sum of the following three matrix products:
\begin{itemize}
\item $M_{p, p^2}(M^\omega)_p$ (i.e., changing the contents of the pushdown tape from $p$ to $pp$ and starting the infinite computations with the leftmost $p$; the second $p$ is never read),
\item $M_{p, p^2} (M^*)_{p, \varepsilon} (M^\omega)_p$ (i.e., changing the contents of the pushdown tape from $p$ to $pp$, emptying the leftmost $p$ by finite computations and starting the infinite computations with the rightmost $p$),
\item $M_{p,p} (M^\omega)_p$ (i.e., changing the contents of the pushdown tape from $p$ to $p$ and starting the infinite computations with this $p$).
\end{itemize}
The forthcoming Theorem \ref{thm:2.7} has an analogous intuitive interpretation.

\begin{thm}
	Let $(\Sr, V)$ be a complete semiring-semimodule pair and let 
	$M \in \roc$ be a roc-matrix. Then
	\begin{equation*}
	(M^\omega)_p \ = \ (M_{p,p^2} + M_{p,p^2}(M^*)_{p,\epsilon} + M_{p,p})(M^\omega)_p \, .
	\end{equation*}
\end{thm}

\begin{proof}
	
	We obtain, by Lemma \ref{lem:2.3}
	\begin{align*}
	& (M_{p,p^2} + M_{p,p^2}(M^*)_{p,\epsilon} + M_{p,p})(M^\omega)_p 	\\
	= \quad & M_{p,p^2}((M^\omega)_p +(M^*)_{p,\epsilon}(M^\omega)_p)
	+ M_{p,p}(M^\omega)_p \\
	= \quad & M_{p,p^2}(M^\omega)_{p^2} + M_{p,p}(M^\omega)_p \ = \ (MM^\omega)_p \ = \ (M^\omega)_p \, .
	\end{align*}
\end{proof}

\begin{cor}
	Let $M \in \roc$ be a roc-matrix.
	Then $(M^\omega)_p$ is a solution of 
	\begin{equation*}
	z = (M_{p,p^2} + M_{p,p^2}(M^*)_{p,\epsilon}+M_{p,p})z \, .
	\end{equation*}
\end{cor}

When we say ``$G$ is the graph with matrix 
$M \in \roc$'' then it means that $G$ is the graph with 
adjacency matrix 
$M' \in \Sr^{(p^*\times n)\times(p^*\times n)}$, where $M'$ 
corresponds to $M$ with respect to the canonical isomorphism 
between $\roc$ and $\Sr^{(p^*\times n) \times (p^*\times n)}$.

Let now $M$ be a roc-matrix and $0\leq k \leq n$.
Then $M^{\omega,k}$ is the column vector in $(V^n)^{p^*}$ 
defined as follows: 
For $i\geq 1$ and $1 \leq j \leq n$, 
let $((M^{\omega,k})_{p^i})_j$ be the sum of all weights of paths in the graph with 
matrix $M$ that have initial vertex $(p^i,j)$ and visit vertices $(p^{i'},j')$, $i'\in \N$, $j' \in \{1, \ldots, k\}$, infinitely often.  
Observe that $M^{\omega, 0}= 0$ and $M^{\omega, n}=M^\omega$. Later on it will be seen that this formalizes the B\"uchi acceptance condition with repeated states $\{1, \ldots, k\}$.

Let $P_k = \{ (j_1, j_2, \dots) \in \{1, \dots, n\}^\omega \mid j_t \leq k \text{ for infinitely many } t \geq 1 \}$.
Then for $1\leq j \leq n$, we obtain
\begin{equation*}
((M^{\omega,k})_p)_j = 
\sum_{i_1,i_2, \dots \geq 1} \sum_{(j_1, j_2, \dots) \in P_k}
(M_{p,p^{i_1}})_{j,j_1}(M_{p^{i_1},p^{i_2}})_{j_1,j_2}(M_{p^{i_2},p^{i_3}})_{j_2,j_3} \dots \, .
\end{equation*}
%

By Theorem 5.4.1 of \'Esik, Kuich \cite{MAT}, we obtain for a finite matrix $A \in \Snn$ and for 
$0 \leq k \leq n$, $1 \leq j \leq n$,

\begin{equation*}
(A^{\omega,k})_j = \sum_{(j_1,j_2, \dots) \in P_k} A_{j,j_1} A_{j_1,j_2}A_{j_2,j_3} \dots \, .
\end{equation*}
Observe that again $A^{\omega, 0}= 0$ and $A^{\omega, n}= A^\omega$.

In the next lemma, we use the following summation identity:
Assume that $A_1, A_2, \dots$ are matrices in $\Snn$. Then for $0 \leq k \leq n$, $1 \leq j \leq n$, and $m \geq 1$, 
\begin{align*}
&\sum_{(j_1,j_2, \dots) \in P_k} (A_1)_{j,j_1}(A_2)_{j_1,j_2\dots} =\\
&\sum_{1 \leq j_1, \dots, j_m \leq n} (A_1)_{j, j_1}\dots
(A_m)_{j_{m-1},j_m} \sum_{(j_{m+1},j_{m+2}, \dots)\in P_k}
(A_{m+1})_{j_m,j_{m+1}} \dots \, .
\end{align*}

\begin{lem}\label{lem:2.6}
	Let $(\Sr, V)$ be a complete semiring-semimodule pair. Let $M \in (\Snn)^{\Gamma^*\times \Gamma^*}$
	be a roc-matrix and $0\leq k \leq n$. Then
	\begin{equation*}
	(M^{\omega,k})_{p^2} \ = \ (M^{\omega,k})_p + (M^*)_{p,\epsilon}(M^{\omega,k})_p \, .
	\end{equation*}
	
\end{lem}

\begin{proof}
	We use the proof of Lemma \ref{lem:2.3}, i.e., the proof for the case $M^{\omega,n} = M^\omega$.
	For $1 \leq j \leq n$, we obtain $	((M^{\omega, k})_{p^2})_j = $
	\begin{align*}
	& \sum_{i_1,i_2, \dots \geq 2} \sum_{(j_1,j_2, \dots)\in P_k}(M_{p,p^{i_1-1}})_{j,j_1}(M_{p^{i_1-1},p^{i_2-1}})_{j_1,j_2} \dots + \\
	& \Big( \sum_{1 \leq j' \leq n} \sum_{m\geq 0}\sum_{k_1,k_2,\dots,k_m \geq 2} \sum _{1 \leq j_1, \dots, j_{m}\leq n} \! \!
	(M_{p,p^{k_1-1}})_{j,j_1} \dots (M_{p^{k_m-1}, \epsilon})_{j_m,j'}\Big) \ \cdot\\
	& \Big(\sum_{k_{m+2},k_{m+3}, \dots\geq 1} 
	\sum_{(j_{m+2},j_{m+3}, \dots)\in P_k}  \!\!
	(M_{p,p^{k_{m+2}}})_{j',j_{m+2}}
	(M_{p^{k_{m+2}},p^{k_{m+3}}})_{j_{m+2},j_{m+3}} \dots\Big) =\\
	& (\Momk_p)_j + \sum_{1 \leq j' \leq n} ((M^*)_{p,\epsilon})_{j,j'} (\Momk_p)_{j'} = \\			
	& ((M^{\omega,k})_p)_j + 
	((M^*)_{p,\epsilon}(M^{\omega,k})_{p})_j =\\
	& ((M^{\omega, k})_p + (M^*)_{p,\epsilon}(M^{\omega,k})_p)_j \, .
	\end{align*}
\end{proof}

\begin{thm}\label{thm:2.7}
	Let $(\Sr, V)$ be a complete semiring-semimodule pair and let 
	$M \in \roc$ be a roc-matrix. Then
	\begin{equation*}
	(M^{\omega,k})_p \ = \ (M_{p,p^2} + M_{p,p^2}(M^*)_{p,\epsilon} + M_{p,p})(M^{\omega,k})_p \, ,
	\end{equation*}
	for all $0 \leq k \leq n$.
\end{thm}

\begin{proof}
	
	We obtain, by Lemma \ref{lem:2.6}, for all $0 \leq k \leq n$,
	\begin{align*}
	& (M_{p,p^2} + M_{p,p^2}(M^*)_{p,\epsilon} + M_{p,p})(M^{\omega,k})_p 	\\
	= \quad & M_{p,p^2}((M^{\omega,k})_p +(M^*)_{p,\epsilon}(M^{\omega,k})_p)
	+ M_{p,p}(M^{\omega,k})_p \\
	= \quad & M_{p,p^2}(M^{\omega,k})_{p^2} + M_{p,p}(M^{\omega,k})_p \ = \ (MM^{\omega,k})_p \ = \ (M^{\omega,k})_p \, .
	\end{align*}
\end{proof}

\begin{cor}\label{cor:2.8}
	Let $M \in \roc$ be a roc-matrix.
	Then, for all $0 \leq k \leq n$,  $(M^{\omega,k})_p$ is a solution of 
	\begin{equation*}
	z = (M_{p,p^2} + M_{p,p^2}(M^*)_{p,\epsilon}+M_{p,p})z \, .
	\end{equation*}
\end{cor}

\section{\texorpdfstring{$\omega$}{Omega}-restricted one-counter automata}
\label{sec:3}

In this section, we define $\omega$-roc automata as a special case of
$\omega$-pushdown automata.
We show that for an $\omega$-roc automaton $\Cmc$
there exists an algebraic system over a complete
semiring-semimodule pair such that the behavior $\bhvr{\Cmc}$ of $\Cmc$ is a component of a solution of this system.

In the sequel, $(\Sr, V)$ is a complete semiring-semimodule pair and $\Sr'$ is a subset of $\Sr$ containing $0$ an $1$.
An \emph{$\Sr'$-$\omega$-pushdown automaton}
$$\Pmc = (n,\Gamma,I,M,P,p_0,k)$$
is given by
\begin{enumerate}[label=(\roman*)]
	\item a finite set of \emph{states} $\{1,\dots,n\}$, 
	$n\geq 1$,
	\item an alphabet $\Gamma$ of \emph{pushdown symbols},
	\item a \emph{pushdown transition matrix} 
	$M \in  (\Spnn)^{\Gamma^* \times \Gamma^*}$,
	\item an \emph{initial state vector} $I \in {S'}^{1 \times n}$,
	\item a \emph{final state vector} $P \in {S'}^{n \times 1}$,
	\item an \emph{initial pushdown symbol} $p_0 \in \Gamma$,
	\item a set of \emph{repeated states} $\{1,\dots,k\}$, $0 \leq k \leq n$.
\end{enumerate}

The definition of a pushdown transition matrix is given
at the beginning of Section \ref{sec:2}. (See also Kuich, Salomaa \cite{88}, Kuich \cite{78} and
\'Esik, Kuich \cite{MAT}.)
Clearly, any roc-matrix is a pushdown transition matrix.

The \emph{behavior of $\Pmc$} is an element of 
$\Sr \times V$ and is defined by 
\begin{equation*}
\bhvr{\Pmc} = (I (M^*)_{p_0,\epsilon}P, I(M^{\omega,k})_{p_0}) \, .
\end{equation*}

Here $I(M^*)_{p_0,\epsilon}P$ is the behavior of the 
$\Sr'$-$\omega$-pushdown automaton
$\Pmc_1=(n,\Gamma,I,M,P,p_0,0)$ and 
$I(M^{\omega,k})_{p_0}$ is the behavior of the $\Sr'$-$\omega$-pushdown automaton 
$\Pmc_2= (n,\Gamma,I,M,0,p_0,k)$.
Observe that $\Pmc_2$ is an automaton with the 
B\"{u}chi acceptance condition: 
if $G$ is the graph with adjacency matrix $M$, then only paths that visit 
the repeated states ${1,\dots,k}$ infinitely often contribute to 
$\bhvr{\Pmc_2}$.
Furthermore, $\Pmc_1$ contains no repeated states and behaves like an ordinary $\Sr'$-pushdown automaton.

An $\Sr'$-$\omega$-roc automaton is an $\Sr'$-$\omega$-pushdown automaton with just one pushdown
symbol such that its pushdown matrix is a roc-matrix.

In the sequel, an $\Sr'$-$\omega$-roc automaton 
$\Pmc=(n,\{p\},I,M,P,p,k)$ is denoted by 
$\Cmc = (n,I,M,P,k)$ with behavior
\begin{equation*}
\bhvr{\Cmc} = (I (M^*)_{p,\epsilon}P, I(M^{\omega,k})_{p}) \, .
\end{equation*}

\begin{rem}
	Consider an $\Sr'$-$\omega$-pushdown automaton $\Pmc$
	with just one pushdown symbol.
	By the construction in the proof of Theorem 13.28 of
	Kuich, Salomaa \cite{88}, an $\Sr'$-$\omega$-roc
	automaton $\Cmc$ can be constructed such that 
	$\bhvr{\Cmc}=\bhvr{\Pmc}$. 
\end{rem}

The next definitions and results are taken from
\'Esik, Kuich \cite[Section 5.6]{MAT}.
For the definition of an $\Sr'$-algebraic system over a quemiring  $S \times V$
we refer the reader to \cite{MAT}, page 136, and for the definition of
quemirings to \cite{MAT}, page 110. 
Here we note that a quemiring   $T$
is isomorphic to a quemiring $S \times V$ determined by the semiring-semimodule pair  $(S,V)$,
cf. \cite{MAT}, page 110.

Observe that the forthcoming system \eqref{sys:1} is a system over the quemiring $\Snn \times V^n$.
Compare the forthcoming algebraic system \eqref{sys:2}
with the algebraic systems occurring in the proofs of 
Theorem 14.15 of Kuich, Salomaa \cite{88}
and Theorem 6.4 of Kuich \cite{78},
both in the case of a roc-matrix.

Let $M$ be a roc-matrix. 
Consider the $\Spnn$-algebraic system over the complete semiring-semimodule pair 
$(\Snn, V^n)$
\begin{equation}
\label{sys:1}
y \ = \ M_{p,p^2} y y + M_{p,p}y + M_{p,\epsilon} \, .
\end{equation}
Then by Theorem 5.6.1 of \'Esik, Kuich \cite{MAT}
$(A, U) \in (\Snn, V^n)$ is a solution of 
\eqref{sys:1} iff $A$ is a solution of the
$\Spnn$-algebraic system over $\Snn$
\begin{equation}
\label{sys:2}
x= M_{p,p^2}xx + M_{p,p}x + M_{p,\epsilon} 
\end{equation}
and $U$ is a solution of the 
$\Snn$-linear system over $V^n$
\begin{equation}
\label{sys:3}
z= M_{p,p^2}z + M_{p,p^2}Az + M_{p,p}z \, . 
\end{equation}

\begin{thm}\label{thm:3.1}
	Let $\Sr$ be a complete starsemiring and $M$ 
	be a roc-matrix.
	Then $(M^*)_{p,\epsilon}$ is a solution of the 
	$\Spnn$-algebraic system \eqref{sys:2}.
	If $S$ is a continuous starsemiring, 
	then $(M^*)_{p,\epsilon}$ is the least solution
	of \eqref{sys:2}.	
\end{thm}

\begin{proof}
	We obtain, by Theorem \ref{thm:2.1} 
	\begin{align*}
	& \Mppsq \Mspe \Mspe + \Mpp \Mspe + \Mpe \\
	= \ & \Mppsq (M^*)_{p^2,\epsilon} + \Mpp \Mspe + \Mpe  \\
	= \ & (MM^*)_{p,\epsilon} = (M^+)_{p,\epsilon} = \Mspe \, .
	\end{align*}
	This proves the first sentence of our theorem.
	The second sentence of Theorem \ref{thm:3.1} is proved by 
	Theorem 6.4 of Kuich \cite{78}. 
\end{proof}

\begin{thm}\label{thm:3.2}
	Let $(\Sr, V)$ be a complete semiring-semimodule pair and $M$ be a roc-matrix. Then
	\begin{equation*}
	(\Mspe, (M^{\omega,k})_p)\, ,
	\end{equation*}	
	is a solution of the $\Spnn$-algebraic system 
	\eqref{sys:1}, for each \ $0\leq k \leq n$.
\end{thm}

\begin{proof}
	Let $0 \leq k \leq n$, and consider the $\Snn$-linear system
	\begin{equation*}
	z= (\Mppsq+\Mppsq\Mspe + \Mpp) z \, .
	\end{equation*}
	By Corollary \ref{cor:2.8}, $\Momk_p$ is a solution of this system.
	Hence, by Theorem 5.6.1 of \'Esik, Kuich \cite{MAT} (see the remark above)
	and Theorem \ref{thm:3.1},
	$(\Mspe,\Momk_p)$ is a solution of the system
	\eqref{sys:1}.
\end{proof}

Observe that, if $\Sr$ is a continuous semiring,
then the $\Snn$-linear system in the proof of
Theorem \ref{thm:3.2} is in fact an $\mathfrak{Alg}({\Sr'})^{n\times n}$-linear system
(see Kuich \cite[p. 623]{78}).

\begin{thm}\label{thm:3.3}
	Let $(\Sr, V)$ be a complete semiring-semimodule pair and let $\Cmc = (n,I,M,P,k)$ be an $\Sr'$-$\omega$-roc-automaton. 
	Then $(\bhvr{\Cmc},(\Mspe,\Momk_p))$
	is a solution of the $\Spnn$-algebraic system
	\begin{equation}\label{sys:4}
	y_0 = I y P \  ,  \quad y = \Mppsq yy + \Mpp y + \Mpe 
	\end{equation} 
	over the complete semiring-semimodule pair $(\Snn,V^n)$.
\end{thm}

\begin{proof}
	By Theorem \ref{thm:3.2}, $(\Mspe, \Momk_p)$ is a solution of the second equation.
	Since 
	\begin{equation*}
	I(\Mspe, \Momk_p)P = (I\Mspe P, I \Momk_p) = \bhvr{\Cmc} \, ,
	\end{equation*}	
	$(\bhvr{\Cmc}, (\Mspe,\Momk_p))$ is a solution of the
	given $\Spnn$-algebraic system.	
\end{proof}	

Let now $\Sr$ be a complete star-omega semiring and $\Sigma$ be an alphabet.
Then by Theorem 5.5.5 of \'Esik, Kuich \cite{MAT},
$(\Sr\ll\! \Sigma^*\! \gg, \Sr \ll \! \Sigma^\omega\! \gg )$ is a complete semiring-semimodule pair.

Let $\Cmc\!  =  \!(n, I, M, P, k)$ be an 
$\Sr \langle \Sigma  \cup  \{\epsilon\}  \rangle $-$\omega$-roc automaton.
Consider the algebraic system \eqref{sys:4}
over the complete semiring-semimodule pair 
${((\Sr\ll \! \Sigma^*\! \gg)^{n \times n},} \allowbreak {(\Sr \ll \! \Sigma^\omega \! \gg)^n )}$
and the mixed algebraic system \eqref{sys:5}
over ${((\Sr\ll \Sigma^* \gg)^{n \times n},}$ ${(\Sr \ll \Sigma^\omega \gg)^n )}$ induced by \eqref{sys:4}
\begin{equation}\label{sys:5}
\begin{alignedat}{3}
x_0 & = I x P, \qquad & x & = \Mppsq xx + \Mpp x + \Mpe \, ,  \\
z_0 & = I z, & z & = \Mppsq z + \Mppsq xz + \Mpp z \, .
\end{alignedat}
\end{equation}
Then, by Theorem \ref{thm:3.3},
\begin{equation*}
(I\Mspe P , \Mspe, I\Momk_p, \Momk_p) , \ 0\leq k \leq n,
\end{equation*}
is a solution of \eqref{sys:5}.
It is called \emph{solution of order $k$}.
Hence, we have proved the next theorem.

\begin{thm}\label{thm:3.4}
	Let $\Sr$ be a complete star-omega 
	semiring and $\Cmc = (n,I,M,P,k)$ be an
	$\Sr \langle \Sigma \cup \{\epsilon\} \rangle$-$\omega$-roc automaton.
	Then $(I\Mspe P , \Mspe, I\Momk_p, \Momk_p) , \ 0\leq k \leq n$,
	is a solution of the mixed algebraic system \eqref{sys:5}.
	\qed
\end{thm}
Let now in \eqref{sys:5}
\begin{equation*}
x = ( [i,p,j])_{1 \leq i,j \leq n}
\end{equation*}
be an $n \times n$-matrix of variables and
\begin{equation*}
z = ([i,p])_{1 \leq i \leq n}
\end{equation*}
be an $n$-dimensional column vector of variables.
If we write the mixed algebraic system \eqref{sys:5}
component-wise, we obtain a mixed algebraic system over
$({(\Sr\ll \Sigma^* \gg),} {(\Sr \ll \Sigma^\omega \gg)} )$
with variables $\ipj$ over $\Sr\llss$, where $1 \leq i,j \leq n$, and
variables $[i,p]$ over $\Sr \ll \Sigma^\omega \gg$,
where $1 \leq i \leq n$.
Observe that we do not really need $p$ in the notation of the variables.
But we want to save the form of the triple construction
in connection with pushdown automata.

Let 
\begin{equation*}
\Mppsq = (a_{ij})_{1 \leq i,j, \leq n}, 
\Mpp = (c_{ij})_{1 \leq i,j, \leq n},
\Mpe = (b_{ij})_{1 \leq i,j, \leq n}
\end{equation*}
and write \eqref{sys:5} with the matrices $x$ and $z$ of variables component-wise then we obtain:
\begin{equation}
\label{sys:6}
\begin{alignedat}{4}
x_0 \ &  = & & \sum_{1 \leq m_1, m_2 \leq n} I_{m_1} 
[m_1,p,m_2] P_{m_2}\\
[i,p,j]\ &  = & & \sum_{1 \leq m_1, m_2 \leq n}  a_{im_1}
[m_1,p,m_2] [m_2,p,j] +\\
& & & \sum_{1\leq m \leq n} c_{im}[m,p,j] + b_{ij}\\
z_0 \ &  = && \sum_{1 \leq m \leq n} I_m [m,p]\\
[i,p] \ & = && \sum_{1 \leq m \leq n} a_{im} [m,p] +
\sum_{1 \leq m_1, m_2 \leq n} a_{im_1} [m_1,p,m_2]
[m_2,p] + \\
& & & \sum_{1\leq m \leq n} c_{im} [m,p]
\end{alignedat}
\end{equation}
for all $1 \leq i,j \leq n$.


\begin{thm}\label{thm:3.5}
	Let $\Sr$ be a complete star-omega 
	semiring and $\Cmc = (n,I,M,P,k)$ be an
	$\Sr \langle \Sigma \cup \{\epsilon\} \rangle$-$\omega$-roc automaton.
	Then 
	\begin{equation*}
	(\sigma_0,  ( \Mspe)_{ij}, \tau_0, \Momk_p)
	\end{equation*}
	is a solution of the system \eqref{sys:6} with
	$\bhvr{\Cmc} = (\sigma_0, \tau_0)$ .
\end{thm}

\begin{proof}
	By Theorem \ref{thm:3.4}.
\end{proof}

\section{Mixed algebraic systems and mixed context-free grammars}
\label{sec:4}

In this section we associate a mixed context-free 
grammar with finite and infinite derivations to the 
algebraic system \eqref{sys:6}.
The language generated by this mixed context-free 
grammar is then the behavior $\bhvr{\Cmc}$ of the
$\omega$-roc automaton $\Cmc$.
The construction of the mixed context-free grammar from
the $\omega$-roc automaton $\Cmc$ is a generalization of 
the well-known triple construction in case of roc 
automata and is called now
\emph{triple-pair construction for $\omega$-roc 
	automata}.
We will consider the commutative complete star-omega semirings 
$\B = ( \{0,1\}, \vee, \land, *,0,1)$ with
$0^* =1^*=1$ and $\N^\infty = (\N \cup \{\infty\}, +, \cdot, ^*, 0,1)$ 
with $0^* =1$ and $a^* = \infty$ for $a \neq \infty$.

If $\Sr = \mathbb{B}$ or $\Sr = \N^\infty$ and $1 \leq k \leq n$, then we associate to the mixed algebraic system \eqref{sys:6} over
$((\Sr\ll \Sigma^* \gg), (\Sr \ll \Sigma^\omega \gg) )$
the \emph{mixed context-free grammar}
\begin{equation*}
G_k \ = \ (X,Z,\Sigma, P_X, P_Z, x_0, z_0, k) \, .
\end{equation*}
(See also \'Esik, Kuich \cite[page 139]{MAT}.)
Here 

\begin{enumerate}[label=(\roman{*})]
	\item $X=\{x_0\} \cup \{[i,p,j]\mid 1\leq i,j\leq n\}$ is a set of \emph{variables for finite derivations};
	\item $Z = \{z_0\} \cup \{[i,p] \mid 1 \leq i \leq n\}$ is a set of \emph{variables for infinite derivations};
	\item $\Sigma$ is an alphabet of \emph{terminal symbols};
	\item $P_X$ is a finite set of \emph{productions for finite derivations} given below;
	\item $P_Z$ is a finite set of \emph{productions for infinite derivations} given below;
	\item $x_0$ is the \emph{start variable for finite derivations};
	\item $z_0$ is the \emph{start variable for infinite derivations};
	\item $\{[i,p] \mid 1 \leq i \leq k \}$ is the set of
	\emph{repeated variables for infinite derivations}.
\end{enumerate}
In the definition of $G_k$
the sets $P_X$ and $P_Z$ are as follows:
\begin{align*}
P_X = \  & \{ x_0 \to a \mpm b \mid \\
& \ \ (I_{m_1},a) \cdot (P_{m_2},b) \neq 0, 
a,b \in \Sigma \cup \{\epsilon\}, 1 \leq m_1, m_2 \leq n \}\ \cup\\
& \{\ipj \to a \mpm [m_2,p,j] \mid \\
& \ \	(\aimi,a) \neq 0, 
a \in \Sigma \cup \{\epsilon\}, 1 \leq i,j,m_1,m_2 \leq n\} \ \cup \\
& \{\ipj \to a \mpj \mid (\cim,a) \neq 0, a \in \Sigma \cup \{\epsilon\},
1 \leq i,j,m \leq n \} \ \cup \\
& \{ \ipj \to a \mid (\bij,a) \neq 0, a \in \Sigma \cup \{\epsilon\},
1 \leq i,j \leq n \} \, ,\\
P_Z = \ & \{z_0 \to a \mp \mid (I_m,a) \neq 0, 
a \in \Sigma \cup \{\epsilon\}, 1 \leq m \leq n \} \ \cup\\
& \{ \ip \to a \mp \mid (\aim,a) \neq 0, a \in \Sigma \cup \{\epsilon\},
1 \leq i,m \leq n \} \ \cup \\
& \{ \ip \to a \mpm [m_2,p] \mid \\
& \ \ (\aimi,a) \neq 0, 
a \in \Sigma \cup \{\epsilon\}, 1 \leq i, m_1, m_2 \leq n \} \ \cup \\
& \{ \ip \to a \mp \mid (\cim,a) \neq 0, 
a \in \Sigma \cup \{\epsilon\}, 1 \leq i,m \leq n \} \, . 
\end{align*}

A \emph{finite leftmost derivation} 
$\alpha_1 \derls \alpha_2$, where
$\alpha_1, \alpha_2 \in (X \cup \Sigma)^*$, by productions
in $P_X$ is defined as usual.
An \emph{infinite (leftmost) derivation} 
$\pi : z_0 \derl^\omega w$, for $z_0 \in Z, w \in \Sigma^\omega$, is defined as follows:
\begin{align*}
\pi:\ & z_0 \derl \alpha_0 [{i_0},p] \derls w_0  [i_0,p] \derl w_0 
\alpha_1 [i_1,p] \derls w_0 w_1  [{i_1},p]  \derl \dots \\
& \derls w_0w_1 \dots w_m  [{i_m},p]  \derl 
w_0 w_1 \dots w_m \alpha_{m+1}
[{i_{m+1}},p]  \derls \dots \, ,
\end{align*}
where $z_0 \to \alpha_0 [{i_0},p] , 
[{i_0},p] \to \alpha_1  [{i_1},p] , \dots, [{i_m},p] \to \alpha_{m+1}[{i_{m+1}},p], \dots$
are productions in $P_Z$ and $w = w_0 w_1 \dots w_m \dots$.

We now define an infinite derivation
$\pi_k : z_0 \derlom w$ for $0 \leq k \leq n$, $z_0 \in Z$, $w \in \Sigma^\omega$:
We take the above definition $\pi:z_0 \Rightarrow^\omega w$ and consider the sequence of the first elements of the variables of $X$ that are rewritten in the finite leftmost derivation $\alpha_m \Rightarrow_L^* w_m$, $m \geq 0$.
Assume this sequence is $i_m^1, i_m^2, \dots, i_m^{t_m}$ for some $t_m$, $m \geq 1$.
Then, to obtain $\pi_k$ from $\pi$, the condition
$i_0, i_1^1, \dots, i_1^{t_1}, i_1, i_2^1, \dots, i_2^{t_2}, i_2, \dots, i_m, i_{m+1}^1, \dots, i_{m+1}^{t_{m+1}}, i_{m+1}, \dots \in P_k$ has to be satisfied.

Then 
\begin{equation*}
L(G_k) =  \{w \in \Sigma^* \mid x_0 \derls w \} \ \cup \ \{ w \in \Sigma^\omega \mid \pi : z_0 \derlom w \} \, . 
\end{equation*}
Observe that the construction of $G_k$ from $\Cmc$ is 
nothing else than a generalization of the triple 
construction in the case of a roc-automaton, if $\Cmc$
is viewed as a pushdown automaton, since the construction
of the context-free grammar $G = (X, \Sigma, P_X,x_0)$
 is the triple construction.
(See Harrison \cite{63}, Theorem 5.4.3; 
Bucher, Maurer \cite{17}, S\"atze 2.3.10, 2.3.30; Kuich, Salomaa \cite{88}, pages 178, 306; Kuich \cite{78}, page 642; 
\'Esik, Kuich \cite{MAT}, pages 77, 78.)

We call the construction of the mixed context-free grammar
$G_k$,  for $0\leq k \leq n$, from $\Cmc$
the \emph{triple-pair construction for $\omega$-roc automata}.
This is justified by the definition of the sets of variables
$\{ \ipj \mid 1 \leq i,j, \leq n\}$ and 
$\{[i,p] \mid 1 \leq i \leq n\}$ of $G_k$ and by the forthcoming Corollary \ref{cor:4.2}.

In the next theorem we use the isomorphism
between ${\B \llss } \times {\B \llso}$ and 
$2^{\Sigma^*} \times 2^{\Sigma^\omega}$.

\begin{thm}\label{thm:4.1}
	Assume that $(\sigma, \tau)$ is the solution of order
	$k$ of the mixed algebraic system \eqref{sys:6} over 
	$(\B \llss , \B \llso)$ for $k \in \{0, \dots, n\}$.
	Then 
	\begin{equation*}
	L(G_k) \ = \ \sigma_{x_0} \cup \tau_{z_0} \, .
	\end{equation*}
\end{thm}

\begin{proof}
	By Theorem \rom{4}.1.2 of Salomaa, Soittola \cite{SalSoi} and by Theorem \ref{thm:3.5}, 
	we obtain $\sigma_{x_0} = \{w \in \Sigma^* \mid x_0 \derls w\}$.
	We now show that $\tau_{z_0}$ is generated by the infinite derivations $\derlom$ from $z_0$.
	First observe that the rewriting by the 
	typical $[i,p,j]$- and $[i,p]$- production 
	corresponds to the situation that 
	in the graph of the $\omega$-restricted one 
	counter automaton $\Cmc$ the edge from $(p\rho,i)$ 
	to $(pp\rho,j),(p\rho,j)$ or $(\rho,j)$, $\rho = p^t$ for some $t \geq 0$ is passed after the 
	state $i$ is visited.
	The first step of the infinite derivation 
	$\pi_k$ is given by $z_0 \derl \alpha_0 [i_0,p]$ and indicates that the path in the graph of $\Cmc$ corresponding to $\pi_k$ starts in state $i_0$.
	Furthermore, the sequence of the first elements of variables that are rewritten in $\pi_k$, i.e., $i_0, i_1^1, \dots, i_1^{t_1}, i_1, i_2^2, \dots, i_2^{t_2}, i_2, \dots, i_m, i_{m+1}^1, \dots, i_{m+1}^{t_{m+1}}, i_{m+1}, \dots$ indicates that the path in the graph of $\Cmc$ corresponding to $\pi_k$ visits these states.
	Since this sequence is in $P_k$ the corresponding path 
	contributes to $\bhvr{\Cmc}$. Hence, by Theorem \ref{thm:3.5}
	we obtain
	\begin{equation*}
	\tau_{z_0} = \{ w \in \Sigma^\omega \mid
	\pi: z_0 \derls w\} \, .
	\end{equation*}
\end{proof}

\begin{cor}\label{cor:4.2}
	Assume that, for some $k \in \{0, \dots, n\}$,
	the mixed context free grammar $G_k$ associated to the mixed algebraic system \eqref{sys:6} is constructed from the 
	$\B \langle \Sigma \cup \{\epsilon\} \rangle$-%
	$\omega$-roc automaton $\Cmc$.
	Then 
	\begin{equation*}
	L(G_k) = \bhvr{\Cmc} \, .	
	\end{equation*}
\end{cor}
\begin{proof}
	By Theorems \ref{thm:3.5} and \ref{thm:4.1}.
\end{proof}

For the remainder of this section our basic semiring is $\N^\infty$,
which allows us to draw some stronger conclusions.

\begin{thm}
	Assume that $(\sigma, \tau)$ is the 
	solution of order $k$ of the mixed algebraic system
	\eqref{sys:6} over
	$(\N^\infty \llss , \N^\infty \llso)$, 
	$k \in \{0, \dots, n\}$, where $I_{m_1}, P_{m_1}, a_{m_1m_2},
	b_{m_1m_2}, c_{m_1m_2}$, $1 \leq m_1,m_2 \leq n$ are in
	$\{0,1\} \langle \Sigma \cup \{\epsilon\}\rangle$.
	Denote by $d(w)$, for $w \in \Sigma^*$, the number 
	(possibly $\infty$) of distinct finite leftmost 
	derivations of $w$ from $x_0$ with respect to $G_k$;
	and by $c(w)$, for $w \in \Sigma^\omega$, the number 
	(possibly $\infty$) of  distinct infinite leftmost
	derivations $\pi$  of $w$ 
	from $z_0$ with respect to $G_k$.
	Then
	\begin{equation*}
	\sigma_{x_0} = \sum_{w \in \Sigma^*} d(w)w \qquad \text{\ and \ } \qquad \tau_{z_0} = \sum_{w \in \Sigma^\omega} c(w)w \, .
	\end{equation*}
\end{thm}

\begin{proof}
	By Theorem \rom{4}.1.5 of Salomaa, Soittola \cite{SalSoi}, Theorems 
	5.5.9 and 5.6.3 of \'Esik, Kuich \cite{MAT} and Theorem \ref{thm:3.5}.
\end{proof}

In the forthcoming Corollary \ref{cor:4.4} we consider,
for a given $\{0,1\}\langle \Sigma \cup \{\epsilon \}\rangle$-$\omega$-roc automaton 
$\Cmc = (n, I, M, P, k)$ the number of distinct computations from an initial instantaneous description $(i,w,p)$ for $w \in \Sigma^*$, $I_i \neq 0$, to an
accepting instantaneous description $(j, \epsilon, \epsilon)$, 
with $P_j \neq 0$, $i,j \in \{0, \dots, n\}$.

Here $(i,w,p)$ means that $\Cmc$ starts in the initial state $i$ with $w$ on its input tape and $p$ on its
pushdown tape;
and $(j,\epsilon,\epsilon)$ means that $\Cmc$ has entered the final state $j$ with empty input tape and empty pushdown tape.

Furthermore, we consider the number of distinct infinite computations starting in an initial instantaneous description 
$(i,w,p)$ for $w \in \Sigma^\infty$, $I_i \neq 0$.

\begin{cor}\label{cor:4.4}
	Assume that, for some $k \in \{0, \dots, n\}$,
	the mixed context-free grammar $G_k$ associated to the mixed algebraic system \eqref{sys:6} is constructed from the $\{0,1\}\langle \Sigma \cup \{\epsilon \}\rangle$-$\omega$-roc automaton $\Cmc$.
	Then the number (possibly $\infty$) of distinct finite leftmost derivations of $w \in \Sigma^*$ from $x_0$
	equals the number of distinct finite computations from 
	an initial instantaneous description for $w$ to an accepting instantaneous description;
	moreover, the number  (possibly $\infty$) of distinct infinite (leftmost) derivations of $w \in \Sigma^\omega$ from $z_0$ equals the number of distinct
	infinite computations starting in an initial instantaneous description for $w$.
\end{cor}

\begin{proof}
	By Corollary 3.4.12 of \'Esik, Kuich \cite[Theorem 4.3]{MAT} and the definition of infinite derivations with respect to $G_k$.
\end{proof}

The context-free grammar $G_k$ associated to \eqref{sys:6}
is called \emph{unambiguous} if each $w \in L(G)$,
$w \in \Sigma^*$ has a unique finite leftmost derivation
and each $w \in L(G)$, $w \in \Sigma^\omega$,
has a unique infinite (leftmost) derivation.

An $\Nlse$-$\omega$-roc automaton $\Cmc$ is called
\emph{unambiguous} if $(\bhvr{\Cmc},w) \in \{0,1\}$ for each $w \in \Sigma^* \cup \Sigma^\omega$.

\begin{cor}
	Assume that, for some $k \in \{0, \dots, n\}$,
	the mixed context-free grammar $G_k$ associated to the mixed algebraic system \eqref{sys:6} is constructed from the  $\{0,1\}\langle \Sigma \cup \{\epsilon \}\rangle$-$\omega$-roc automaton $\Cmc$.
	Then $G_k$ is unambiguous iff $\bhvr{\Cmc}$
	is unambiguous. 	
\end{cor}

In the forthcoming paper Droste, \'Esik, Kuich \cite{UP} we extend the results of this paper to weighted $\omega$-pushdown automata and obtain the triple-pair construction for them. In the classical theory this triple-pair constructions extends the well-known triple construction that, given an $\omega$-pushdown automaton, yields an equivalent context-free grammar.

\section*{Acknowledgment}
\noindent	The ideas of and personal discussions with Zolt\'an \'Esik were of great influence
	in preparing this paper. Thanks are due to two unknown referees for their helpful remarks.


\begin{thebibliography}{99}
	
	\bibitem{4}
	Berstel, J.: Transductions and Context-Free Languages. Teubner, 1979.
	
	\bibitem{10}
	Bloom, S.~L., \'Esik, Z.:  Iteration Theories. 
	EATCS Monographs on Theoretical Computer Science. Springer, 1993.
	
	\bibitem{17}
	Bucher, W., Maurer, H.: Theoretische Grundlagen der Programmiersprachen.
	B. I. Wissenschaftsverlag, 1984.
	
	\bibitem{25}
	Conway, J.~H.:  Regular Algebra and Finite Machines. Chapman \& Hall,
	1971.
	
	\bibitem{UP}
	Droste, M., \'Esik, Z., Kuich, W.: The triple-pair construction for weighted $\omega$-pushdown automata. In: Automata and Formal Languages (AFL 2017), EPTCS (2017) 101-113.
	
	\bibitem{29} Eilenberg, S.: Automata, Languages and Machines. Vol. A. Academic
	Press, 1974.
	
	\bibitem{42}
	\'Esik, Z., Kuich, W.:   A semiring-semimodule generalization of $\omega$-context-free languages. In: Theory is Forever (Eds.: J. Karhum{\"a}ki, H. Maurer, G. Paun, G. Rozenberg), LNCS 3113, Springer, 2004, 68--80.
	
	\bibitem{43}
	\'Esik,  Z.,  Kuich,  W.:  A  semiring-semimodule  generalization  of
	$\omega$-regular languages II. 
	Journal of Automata, Languages and Combinatorics 10 (2005) 243--264.
	
	\bibitem{44}
	\'Esik, Z., Kuich, W.: On iteration semiring-semimodule pairs. 
	Semigroup Forum 75 (2007), 129--159.
	
	\bibitem{45}
	\'Esik, Z., Kuich, W.:  A semiring-semimodule generalization of transducers
	and abstract $\omega$-families of power series. 
	Journal of Automata, Languages and Combinatorics,
	12 (2007), 435--454.
	
	\bibitem{MAT}
	\'Esik, Z., Kuich, W.: Modern Automata Theory.
	\url{http://www.dmg.tuwien.ac.at/kuich}
	
	\bibitem{EK}
	\'Esik, Z., Kuich, W.: Continuous semiring-semimodule pairs and mixed algebraic systems.
	Acta Cybernetica 252 (2017) 43-59.
	
	\bibitem{JACM16}
	Greibach S.~A.: An infinite hierarchy of context-free languages. Journal of the ACM 16 (1969) 91--106.
	
	\bibitem{63}
	Harrison, M. A.: Introduction to Formal Language Theory. Addison-Wesley, 1978.
	
	
	\bibitem{78}
	Kuich, W.:  Semirings and formal power series:  
	Their relevance to formal languages and automata theory. 
	In:  Handbook of Formal Languages (Eds.: G. Rozenberg and A. Salomaa), Springer, 1997, Vol. 1, Chapter 9, 609--677.
	
	\bibitem{88}
	Kuich, W., Salomaa, A.:  Semirings, Automata, Languages. 
	EATCS Monographs on Theoretical Computer Science, Vol. 5. Springer, 1986.
	
	\bibitem{PerPin}
	Perrin, D., Pin, J. - E.: Infinite Words – Automata, Semigroups, Logic and Games, Elsevier, 2004.
	
	\bibitem{SalSoi}
	Salomaa, A., Soittola, M.: Automata - Theoretic Aspects of Formal Power Series, Springer, 1978.
	
\end{thebibliography}
\end{document}